\documentclass[fleqn,12pt,twoside]{article}

\usepackage[headings]{espcrc1}
\readRCS
$Id: espcrc1.tex,v 1.2 2004/02/24 11:22:11 spepping Exp $
\usepackage{graphicx}
\usepackage[figuresright]{rotating}

\newtheorem{theorem}{Theorem}

\newtheorem{corollary}[theorem]{Corollary}

\newtheorem{problem}{Problem}

\newenvironment{proof}[1][Proof.]{\begin{trivlist}
\item[\hskip \labelsep {\bfseries #1}]}{\end{trivlist}}

\newenvironment{acknowledgement}[1][Acknowledgement]{\begin{trivlist}
\item[\hskip \labelsep {\bfseries #1}]}{\end{trivlist}}

\newcommand{\AmS}{{\protect\the\textfont2
  A\kern-.1667em\lower.5ex\hbox{M}\kern-.125emS}}

\usepackage{amsmath}
\usepackage{amsfonts}
\usepackage{amssymb}
\title{Interval edge colorings of some products of graphs}

\author{Petros A. Petrosyan\address[MCSD]{Institute for Informatics and Automation Problems, National Academy of
Sciences, 0014, Armenia}%
\address{Department of Informatics and Applied Mathematics, Yerevan State
University, 0025, Armenia} \thanks {email:
pet\_petros@\{ipia.sci.am, ysu.am, yahoo.com\}}}

\runtitle{Interval edge colorings of some products of graphs}
\runauthor{Petros A. Petrosyan}

\begin{document}

\maketitle

\begin{abstract}
An edge coloring of a graph $G$ with colors $1,2,\ldots ,t$ is
called an interval $t$-coloring if for each $i\in \{1,2,\ldots,t\}$
there is at least one edge of $G$ colored by $i$, and the colors of
edges incident to any vertex of $G$ are distinct and form an
interval of integers. A graph $G$ is interval colorable, if there is
an integer $t\geq 1$ for which $G$ has an interval $t$-coloring. Let
$\mathfrak{N}$ be the set of all interval colorable graphs.
 In 2004 Kubale and Giaro showed that if $G,H\in \mathfrak{N}$, then the
 Cartesian product of these graphs belongs to $\mathfrak{N}$.
Also, they formulated a similar problem for the lexicographic
product as an open problem. In this paper we first show that if
$G\in \mathfrak{N}$, then $G[nK_{1}]\in \mathfrak{N}$ for any $n\in
\mathbf{N}$. Furthermore, we show that if $G,H\in \mathfrak{N}$ and
$H$ is a regular graph, then strong and lexicographic products of
graphs $G,H$ belong to $\mathfrak{N}$. We also prove that tensor and
strong tensor products of graphs $G,H$ belong to $\mathfrak{N}$ if
$G\in \mathfrak{N}$ and $H$ is a regular graph.\\


\end{abstract}

\section{Introduction}\

An edge coloring of a graph $G$ with colors $1,2,\ldots ,t$ is
called an interval $t$-coloring if for each $i\in \{1,2,\ldots,t\}$
there is at least one edge of $G$ colored by $i$, and the colors of
edges incident to any vertex of $G$ are distinct and form an
interval of integers. Interval edge colorings naturally arise in
scheduling problems and are related to the problem of constructing
timetables without \textquotedblleft gaps\textquotedblright for
teachers and classes. The notion of interval edge colorings was
introduced by Asratian and Kamalian \cite{b1} in 1987. In \cite{b1}
they proved that if a triangle-free graph $G=\left(V,E\right)$ has
an interval $t$-coloring, then $t\leq \left\vert V\right\vert -1$.
In \cite{b19} interval edge colorings of complete bipartite graphs
and trees were investigated. Furthermore, Kamalian \cite{b20} showed
that if $G$ admits an interval $t$-coloring, then $t\leq 2\left\vert
V\right\vert -3$. Giaro, Kubale and Malafiejski \cite{b12} proved
that this upper bound can be improved to $2\left\vert V\right\vert
-4$ if $\left\vert V\right\vert \geq 3$. For a planar graph $G$,
Axenovich \cite{b5} showed that if $G$ has an interval $t$-coloring,
then $t\leq \frac{11}{6}\left\vert V\right\vert$. In general, it is
an $NP$-complete problem to decide whether a given bipartite graph
$G$ admits an interval edge coloring \cite{b35}. In papers
\cite{b2,b4,b5,b7,b8,b9,b10,b11,b12,b13,b14,b15,b16,b19,b20,b21,b22,b23,b29,b30,b32,b34}
the problem of existence and construction of interval edge colorings
was considered and some bounds for the number of colors in such
colorings of some classes of graphs were given. Surveys on this
topic can be found in some books \cite{b3,b18,b25}.

The different products of graphs were introduced by Berge \cite{b6},
Sabidussi \cite{b33} and Vizing \cite{b36}. There are many papers
\cite{b17,b24,b26,b27,b28,b31,b38} devoted to edge colorings of
various products of graphs. In this paper we investigate interval
edge colorings of various products of graphs.

\bigskip

\section{Definitions and preliminary results}\

All graphs considered in this paper are finite, undirected and have
no loops or multiple edges. Let $V(G)$ and $E(G)$ denote the sets of
vertices and edges of $G$, respectively. The maximum degree of a
vertex of $G$ is denoted by $\Delta (G)$ and the chromatic index of
$G$ by $\chi ^{\prime }\left( G\right)$. A partial edge coloring of
$G$ is a coloring of some of the edges of $G$ such that no two
adjacent edges receive the same color. If $\alpha $ is a partial
edge coloring of $G$ and $v\in V(G)$ then $S\left( v,\alpha \right)$
denotes the set of colors of colored edges incident to $v$.

A graph $G$ is interval colorable, if there is an integer $t\geq 1$,
for which $G$ has an interval $t$-coloring. Let $\mathfrak{N}$ be
the set of all interval colorable graphs \cite{b1,b20}. For a graph
$G\in \mathfrak{N}$, the least and the greatest values of $t$ for
which $G$ has an interval $t$-coloring are denoted by $w\left(
G\right)$ and $W\left(G\right)$, respectively.

Let $G=(V(G),E(G))$ and $H=(V(H),E(H))$ be two graphs.

The Cartesian product $G\square H$ is defined as follows:
\begin{center}
$V(G\square H)=V(G)\times V(H)$, $E(G\square
H)=\{((u_{1},v_{1}),(u_{2},v_{2}))|$

$u_{1}=u_{2}~and~(v_{1},v_{2})\in E(H)~or~v_{1}=v_{2}~and~
(u_{1},u_{2})\in E(G)\}$.
\end{center}

The tensor (direct) product $G\times H$ is defined as follows:
\begin{center}
$V(G\times H)=V(G)\times V(H)$,

$E(G\times H)=\{((u_{1},v_{1}),(u_{2},v_{2}))|~(u_{1},u_{2})\in
E(G)~and~(v_{1},v_{2})\in E(H)\}$.
\end{center}

The strong tensor (semistrong) product $G\otimes H$ is defined as
follows:
\begin{center}
$V(G\otimes H)=V(G)\times V(H)$, $E(G\otimes
H)=\{((u_{1},v_{1}),(u_{2},v_{2}))|$

$~(u_{1},u_{2})\in E(G)~and~(v_{1},v_{2})\in
E(H)~or~v_{1}=v_{2}~and~(u_{1},u_{2})\in E(G)\}$.
\end{center}

The strong product $G\boxtimes H$ is defined as follows:
\begin{center}
$V(G\boxtimes H)=V(G)\times V(H)$, $E(G\boxtimes
H)=\{((u_{1},v_{1}),(u_{2},v_{2}))|~(u_{1},u_{2})\in E(G)$

$~and~(v_{1},v_{2})\in E(H)~or~u_{1}=u_{2}~and~(v_{1},v_{2})\in
E(H)~or~v_{1}=v_{2}~and~ (u_{1},u_{2})\in E(G)\}$.
\end{center}

The lexicographic product (composition) $G[H]$ is defined as
follows:
\begin{center}
$V(G[H])=V(G)\times V(H)$,

$E(G[H])=\{((u_{1},v_{1}),(u_{2},v_{2}))|~(u_{1},u_{2})\in
E(G)~or~u_{1}=u_{2}~and~(v_{1},v_{2})\in E(H)\}$.
\end{center}

The terms and concepts that we do not define can be found in
\cite{b37}.\\

Asratian and Kamalian proved the following:
\begin{theorem}
\label{mytheorem1}\cite{b1}. Let $G$ be a regular graph. Then

\begin{description}
\item[(1)] $G\in \mathfrak{N}$ if and only if $\chi ^{\prime}(G)=\Delta(G)$.

\item[(2)] If $G\in \mathfrak{N}$ and $\Delta(G)\leq t\leq W(G)$, then $G$
has an interval $t$-coloring.
\end{description}
\end{theorem}

\begin{corollary}
\label{mycorollary1} If $G$ is an $r$-regular bipartite graph, then
$G\in \mathfrak{N}$ and $w(G)=r$.
\end{corollary}

Kubale and Giaro proved the following:
\begin{theorem}
\label{mytheorem2}\cite{b25}. If $G,H\in \mathfrak{N}$, then
$G\square H\in \mathfrak{N}$. Moreover, $w(G\square H)\leq
w(G)+w(H)$ and $W(G\square H)\geq W(G)+W(H)$.
\end{theorem}

The $k$-dimensional grid $G(n_{1},n_{2},\ldots,n_{k})$, $n_{i}\in
\mathbf{N}$ is the Cartesian product of paths $P_{n_{1}}\square
P_{n_{2}}\square \cdots \square P_{n_{k}}$. The cylinder
$C(n_{1},n_{2})$ is the Cartesian product $P_{n_{1}}\square
C_{n_{2}}$ and the torus $T(n_{1},n_{2})$ is the Cartesian product
$C_{n_{1}}\square C_{n_{2}}$, where $C_{n_{i}}$ is the cycle of
length $n_{i}$. For these graphs Kubale and Giaro proved the
following:

\begin{theorem}
\label{mytheorem3}\cite{b10}. If $G=G(n_{1},n_{2},\ldots,n_{k})$ or
$G=C(m,2n)$, $m\in \mathbf{N}$, $n\geq 2$, or $G=T(2m,2n)$, $m,n\geq
2$, then $G\in \mathfrak{N}$ and $w(G)=\Delta (G)$.
\end{theorem}

For the greatest possible number of colors in interval edge
colorings of grid graphs Petrosyan and Karapetyan proved the
following theorems:

\begin{theorem}
\label{mytheorem4}\cite{b29}. If $G=C(m,2n)$, $m\in \mathbf{N}$,
$n\geq 2$, then $W(G)\geq 3m+n-2$.
\end{theorem}

\begin{theorem}
\label{mytheorem5}\cite{b29}. If $G=T(2m,2n)$, $m,n \geq 2$, then
$W(G)\geq \max\{3m+n,3n+m\}$.
\end{theorem}

In \cite{b30} Petrosyan investigated interval edge colorings of
complete graphs and $n$-dimensional cubes $Q_{n}$. In particular, he
proved the following theorems:

\begin{theorem}
\label{mytheorem6} $~W\left(Q_{n}\right)\geq
\frac{n\left(n+1\right)}{2}$ for any $n\in \mathbf{N}$.
\end{theorem}

\begin{theorem}
\label{mytheorem7} Let $n=p2^{q}$, where $p$ is odd and $q$ is
nonnegative. Then
\begin{center}
$W\left(K_{2n}\right)\geq 4n-2-p-q$.
\end{center}
\end{theorem}

The Hamming graph $H(n_{1},n_{2},\ldots,n_{k})$, $n_{i}\in
\mathbf{N}$ is the Cartesian product of complete graphs
$K_{n_{1}}\square K_{n_{2}}\square \cdots \square K_{n_{k}}$. The
graph $H_{n}^{k}$ is the Cartesian product of the complete graph
$K_{n}$ by itself $k$ times. It is easy to see that from Theorems
\ref{mytheorem1}, \ref{mytheorem2} and \ref{mytheorem7}, we have the
following result:

\begin{theorem}
\label{mytheorem8} Let $n=p2^{q}$, where $p$ is odd and $q$ is
nonnegative. Then
\begin{description}
\item[(1)] $H_{2n}^{k}\in \mathfrak{N}$,

\item[(2)] $w(H_{2n}^{k})=(2n-1)k$,

\item[(3)] $W(H_{2n}^{k})\geq (4n-2-p-q)k$.
\end{description}
\end{theorem}

It is known that there are graphs $G$ and $H$ for which $G\square
H\in \mathfrak{N}$ ($G[H]\in \mathfrak{N}$), but $G\in
\mathfrak{N}$, $H\notin\mathfrak{N}$ or $G,H\notin\mathfrak{N}$. For
example, $K_{2}\square C_{3}\in \mathfrak{N}$ and $K_{1,1,3}\square
C_{3}\in \mathfrak{N}$ ($K_{2}[C_{5}]\in \mathfrak{N}$ and
$C_{5}[P]\in \mathfrak{N}$), but $K_{1,1,3},C_{3}\notin
\mathfrak{N}$ ($P,C_{5}\notin \mathfrak{N}$, where $P$ is the
Petersen graph). Moreover, general results can be obtained from the
following theorems:

\begin{theorem}
\label{mytheorem9} (Kotzig \cite{b24}, Pisanski, Shawe-Taylor, Mohar
\cite{b31}) If $G$ and $H$ are two regular graphs for which at least
one of the following conditions holds:

\begin{description}
\item[(1)] $G$ and $H$ contain a perfect matching,

\item[(2)] $\chi^{\prime}(G)=\Delta(G)$,

\item[(3)] $\chi^{\prime}(H)=\Delta(H)$,
\end{description}
then $\chi^{\prime}(G\square H)=\Delta(G\square H)$ and
$\chi^{\prime}(G[H])=\Delta(G[H])$.
\end{theorem}

\begin{theorem}
\label{mytheorem10} (Kotzig \cite{b24}, Pisanski, Shawe-Taylor,
Mohar \cite{b31}) Let $G$ be a cubic graph. Then
$\chi^{\prime}(G\square C_{n})=\Delta(G\square C_{n})=5$ and
$\chi^{\prime}(C_{n}[G])=\Delta(C_{n}[G])$ for any $n\geq 4$.
\end{theorem}

\begin{corollary}
\label{mycorollary2} If $G$ and $H$ are two regular graphs for which
at least one of the following conditions holds:

\begin{description}
\item[(1)] $G$ and $H$ contain a perfect matching,

\item[(2)] $G\in \mathfrak{N}$,

\item[(3)] $H\in \mathfrak{N}$,
\end{description}
then $G\square H,G[H]\in \mathfrak{N}$ and $w(G\square
H)=\Delta(G\square H)$, $w(G[H])=\Delta(G[H])$.
\end{corollary}

\begin{corollary}
\label{mycorollary3} Let $G$ be a cubic graph. Then $G\square
C_{n},C_{n}[G]\in \mathfrak{N}$ and $w(G\square
C_{n})=\Delta(G\square C_{n})=5$, $w(C_{n}[G])=\Delta(C_{n}[G])$ for
any $n\geq 4$.
\end{corollary}

\begin{theorem}
\label{mytheorem11} The torus $T(n_{1},n_{2})\in \mathfrak{N}$ if
$n_{1}\cdot n_{2}$ is even, $T(n_{1},n_{2})\notin \mathfrak{N}$ if
$n_{1}\cdot n_{2}$ is odd and the Hamming graph
$H(n_{1},n_{2},\ldots,n_{k})\in \mathfrak{N}$ if $n_{1}\cdot
n_{2}\cdots n_{k}$ is even, $H(n_{1},n_{2},\ldots,n_{k})\notin
\mathfrak{N}$ if $n_{1}\cdot n_{2}\cdots n_{k}$ is odd.
\end{theorem}

\begin{proof} Since $T(n_{1},n_{2})$ and $H(n_{1},n_{2},\ldots,n_{k})$ are regular graphs,
by Theorem \ref{mytheorem1} and Corollary \ref{mycorollary2}, we
have $T(n_{1},n_{2})\in \mathfrak{N}$ when $n_{1}\cdot n_{2}$ is
even and $H(n_{1},n_{2},\ldots,n_{k})\in \mathfrak{N}$ when
$n_{1}\cdot n_{2}\cdots n_{k}$ is even.

Let us show that $T(n_{1},n_{2})\notin \mathfrak{N}$ when
$n_{1}\cdot n_{2}$ is odd and $H(n_{1},n_{2},\ldots,n_{k})\notin
\mathfrak{N}$ when $n_{1}\cdot n_{2}\cdots n_{k}$ is odd.

Since $T(n_{1},n_{2})$ and $H(n_{1},n_{2},\ldots,n_{k})$ are regular
graphs, we have

\begin{center}
$\vert E(T(n_{1},n_{2}))\vert= 2n_{1}\cdot n_{2}$ and $\vert
E(H(n_{1},n_{2},\ldots,n_{k}))\vert = \frac{n_{1}\cdot n_{2}\cdots
n_{k}\cdot \Delta(H(n_{1},n_{2},\ldots,n_{k}))}{2}$.
\end{center}

If $\chi^{\prime}(T(n_{1},n_{2}))=\Delta(T(n_{1},n_{2}))=4$, then
\begin{center}
$\vert E(T(n_{1},n_{2}))\vert \leq 2(n_{1}\cdot n_{2}-1)$, since
$n_{1}\cdot n_{2}$ is odd.
\end{center}

This shows that
$\chi^{\prime}(T(n_{1},n_{2}))=\Delta(T(n_{1},n_{2}))+1=5$ and, by
Theorem \ref{mytheorem1}, $T(n_{1},n_{2})\notin \mathfrak{N}$.

Similarly, if
$\chi^{\prime}(H(n_{1},n_{2},\ldots,n_{k}))=\Delta(H(n_{1},n_{2},\ldots,n_{k}))$,
then
\begin{center}
$\vert E(H(n_{1},n_{2},\ldots,n_{k}))\vert \leq
\frac{\left(n_{1}\cdot n_{2}\cdots n_{k}-1\right)\cdot
\Delta(H(n_{1},n_{2},\ldots,n_{k}))}{2}$, since $n_{1}\cdot
n_{2}\cdots n_{k}$ is odd.
\end{center}

This shows that
$\chi^{\prime}(H(n_{1},n_{2},\ldots,n_{k}))=\Delta(H(n_{1},n_{2},\ldots,n_{k}))+1$
and, by Theorem \ref{mytheorem1}, $H(n_{1},n_{2},\ldots,n_{k})\notin
\mathfrak{N}$. ~$\square$
\end{proof}

\bigskip

\section{Main results}\

First, we consider interval edge colorings of the tensor product of
graphs. In \cite{b25} Kubale and Giaro noted that there are graphs
$G,H\in \mathfrak{N}$, such that $G\times H\notin \mathfrak{N}$.
Here, we prove that if one of the graphs belongs to $\mathfrak{N}$
and the other is regular, then $G\times H\in \mathfrak{N}$.

\begin{theorem}
\label{mytheorem12} If $G\in \mathfrak{N}$ and $H$ is an $r$-regular
graph, then $G\times H\in \mathfrak{N}$. Moreover, $w(G\times H)\leq
w(G)\cdot r$ and $W(G\times H)\geq W(G)\cdot r$.
\end{theorem}
\begin{proof} Let $V(G)=\{u_{1},u_{2},\ldots,u_{n}\}$,
$V(H)=\{v_{1},v_{2},\ldots,v_{m}\}$ and

\begin{center}
$V\left(G\times H\right)=\left\{w_{j}^{(i)}|~1\leq i\leq n, 1\leq
j\leq m\right\}$,
\end{center}

\begin{center}
$E\left(G\times
H\right)=\left\{\left(w_{p}^{(i)},w_{q}^{(j)}\right)|\left(u_{i},u_{j}\right)\in
E(G)~and~\left(v_{p},v_{q}\right)\in E(H)\right\}$.
\end{center}

Let us consider the graph $K_{2}\times H$. Clearly, $K_{2}\times H$
is an $r$-regular bipartite graph, thus, by Corollary
\ref{mycorollary1}, $K_{2}\times H\in \mathfrak{N}$ and
$w\left(K_{2}\times H\right)=r$. Let $\alpha$ be an interval
$t$-coloring of the graph $G$, $\beta$ be an interval $r$-coloring
of the graph $K_{2}\times H$ and

\begin{center}
$V\left(K_{2}\times
H\right)=\{x_{1},x_{2},\ldots,x_{m},y_{1},y_{2},\ldots,y_{m}\}$,
\end{center}

\begin{center}
$E\left(K_{2}\times
H\right)=\left\{\left(x_{i},y_{j}\right)|\left(v_{i},v_{j}\right)\in
E(H), 1\leq i\leq m,1\leq j\leq m\right\}$.
\end{center}

Define an edge coloring $\gamma$ of the graph $G\times H$ in the
following way:

for every $\left(w_{p}^{(i)},w_{q}^{(j)}\right)\in E(G\times H)$

\begin{center}
$\gamma
\left(\left(w_{p}^{(i)},w_{q}^{(j)}\right)\right)=\left(\alpha
\left(\left(u_{i},u_{j}\right)\right)-1\right)\cdot r+\beta
\left(\left(x_{p},y_{q}\right)\right)$,
\end{center}
where $1\leq i\leq n, 1\leq j\leq n, 1\leq p\leq m, 1\leq q\leq
m$.\\

It is not difficult to see that $\gamma$ is an interval $t\cdot
r$-coloring of the graph $G\times H$. By the definition of $\gamma$,
we have $w(G\times H)\leq w(G)\cdot r$ and $W(G\times H)\geq
W(G)\cdot r$. ~$\square$
\end{proof}

The Figure \ref{example} shows the interval $6$-coloring $\gamma$ of
the graph $P_{4}\times C_{5}$ described in the proof of Theorem
\ref{mytheorem12}.\\

\begin{figure}[h]
\begin{center}
\includegraphics[height=20pc,width=25pc]{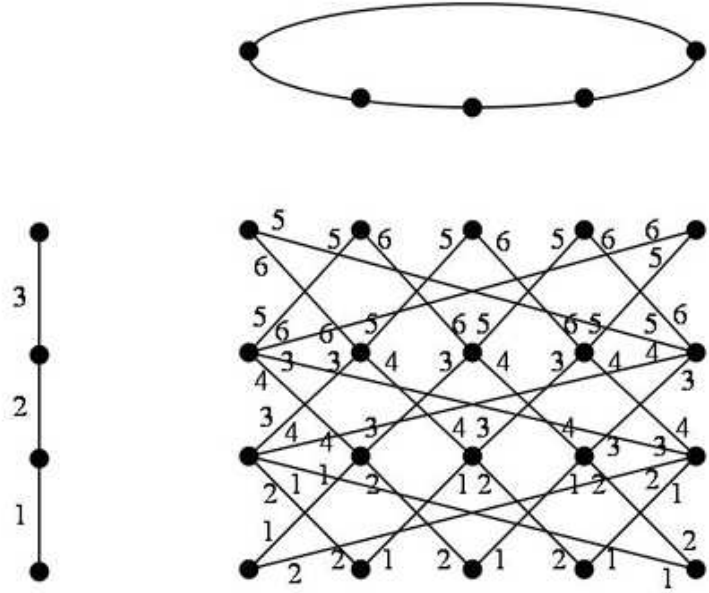}\
\caption{The interval $6$-coloring $\gamma$ of the graph
$P_{4}\times C_{5}$.}\label{example}
\end{center}
\end{figure}

Note that from Theorems \ref{mytheorem1} and  \ref{mytheorem12}, we
have the following result:

\begin{corollary}
\label{mycorollary4} (Pisanski, Shawe-Taylor, Mohar \cite{b31}) If
$G$ is $1$-factorable and $H$ is a regular graph, then $G\times H$
is also $1$-factorable.
\end{corollary}

We showed that if $G\in \mathfrak{N}$ and $H$ is regular, then
$G\times H\in \mathfrak{N}$. Now we prove a similar result for the
strong tensor product of graphs.

\begin{theorem}
\label{mytheorem13} If $G\in \mathfrak{N}$ and $H$ is an $r$-regular
graph, then $G\otimes H\in \mathfrak{N}$. Moreover, $w(G\otimes
H)\leq w(G)\cdot (r+1)$ and $W(G\otimes H)\geq W(G)\cdot (r+1)$.
\end{theorem}
\begin{proof} Let $V(G)=\{u_{1},u_{2},\ldots,u_{n}\}$,
$V(H)=\{v_{1},v_{2},\ldots,v_{m}\}$ and

\begin{center}
$V\left(G\otimes H\right)=\left\{w_{j}^{(i)}|~1\leq i\leq n, 1\leq
j\leq m\right\}$,
\end{center}

\begin{center}
$E\left(G\otimes H\right)=E\left(G\times
H\right)\cup\left\{\left(w_{p}^{(i)},w_{p}^{(j)}\right)|~1\leq p\leq
m~and~\left(u_{i},u_{j}\right)\in E(G)\right\}$.
\end{center}

Let us consider the graph $K_{2}\otimes H$. Clearly, $K_{2}\otimes
H$ is an $(r+1)$-regular bipartite graph, thus, by Corollary
\ref{mycorollary1}, $K_{2}\otimes H\in \mathfrak{N}$ and
$w\left(K_{2}\otimes H\right)=r+1$. Let $\alpha$ be an interval
$t$-coloring of the graph $G$, $\beta$ be an interval
$(r+1)$-coloring of the graph $K_{2}\otimes H$ and

\begin{center}
$V\left(K_{2}\otimes
H\right)=\{x_{1},x_{2},\ldots,x_{m},y_{1},y_{2},\ldots,y_{m}\}$,
\end{center}

\begin{center}
$E\left(K_{2}\otimes H\right)=\left\{\left(x_{i},y_{i}\right)|~1\leq
i\leq m\right\}\cup E\left(K_{2}\times H\right)$.
\end{center}

Define an edge coloring $\gamma$ of the graph $G\otimes H$ in the
following way:

for every $\left(w_{p}^{(i)},w_{q}^{(j)}\right)\in E(G\otimes H)$

\begin{center}
$\gamma
\left(\left(w_{p}^{(i)},w_{q}^{(j)}\right)\right)=\left(\alpha
\left(\left(u_{i},u_{j}\right)\right)-1\right)\cdot (r+1)+\beta
\left(\left(x_{p},y_{q}\right)\right)$,
\end{center}
where $1\leq i\leq n, 1\leq j\leq n, 1\leq p\leq m, 1\leq q\leq
m$.\\

It is not difficult to see that $\gamma$ is an interval $t\cdot
(r+1)$-coloring of the graph $G\otimes H$. By the definition of
$\gamma$, we have $w(G\otimes H)\leq w(G)\cdot (r+1)$ and
$W(G\otimes H)\geq W(G)\cdot (r+1)$. ~$\square$
\end{proof}

The Figure \ref{example1} shows the interval $9$-coloring $\gamma$
of the graph $P_{4}\otimes C_{5}$ described in the proof of Theorem
\ref{mytheorem13}.\\

\begin{figure}[h]
\begin{center}
\includegraphics[height=20pc,width=25pc]{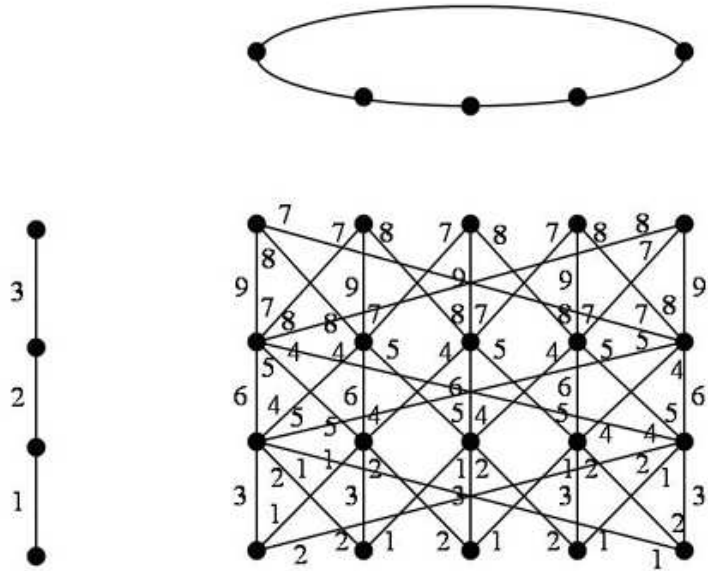}\
\caption{The interval $9$-coloring $\gamma$ of the graph
$P_{4}\otimes C_{5}$.}\label{example1}
\end{center}
\end{figure}

Note that from Theorems \ref{mytheorem1} and \ref{mytheorem13}, we
have the following result:

\begin{corollary}
\label{mycorollary5} (Pisanski, Shawe-Taylor, Mohar \cite{b31}) If
$G$ is $1$-factorable and $H$ is a regular graph, then $G\otimes H$
is also $1$-factorable.
\end{corollary}

Next, we consider interval edge colorings of the strong product of
graphs. In \cite{b25} Kubale and Giaro noted that there are graphs
$G,H\in \mathfrak{N}$, such that $G\boxtimes H\notin \mathfrak{N}$.
Here, we prove that if two graphs belong to $\mathfrak{N}$ and one
of them is regular, then $G\boxtimes H\in \mathfrak{N}$.

\begin{theorem}
\label{mytheorem14} If $G,H\in \mathfrak{N}$ and $H$ is an
$r$-regular graph, then $G\boxtimes H\in \mathfrak{N}$. Moreover,
$w(G\boxtimes H)\leq w(G)\cdot (r+1)+r$ and $W(G\boxtimes H)\geq
W(G)\cdot (r+1)+r$.
\end{theorem}
\begin{proof} Let $V(G)=\{u_{1},u_{2},\ldots,u_{n}\}$,
$V(H)=\{v_{1},v_{2},\ldots,v_{m}\}$ and

\begin{center}
$V\left(G\boxtimes H\right)=\bigcup_{i=1}^{n}V^{i}(H)$, where
$V^{i}(H)=\left\{w_{j}^{(i)}|~1\leq j\leq m\right\}$,
\end{center}

\begin{center}
$E\left(G\boxtimes H\right)=E\left(G\otimes H\right)\cup
\bigcup_{i=1}^{n}E^{i}(H)$, where
\end{center}

\begin{center}
$E^{i}(H)=\left\{\left(w_{p}^{(i)},w_{q}^{(i)}\right)|~\left(v_{p},v_{q}\right)\in
E(H)\right\}$.
\end{center}

For $i=1,2,\ldots,n$, define a graph $H_{i}$ as follows:
\begin{center}
$H_{i}=\left(V^{i}(H),E^{i}(H)\right)$.
\end{center}

First of all note that $\chi ^{\prime}(H)=\Delta(H)=r$ since $H\in
\mathfrak{N}$ and $H$ is an $r$-regular graph. This implies that
there exists an interval $r$-coloring of the graph $H$. Let us
consider the graph $K_{2}\otimes H$. Clearly, $K_{2}\otimes H$ is an
$(r+1)$-regular bipartite graph, thus, by Corollary
\ref{mycorollary1}, $K_{2}\otimes H\in \mathfrak{N}$ and
$w\left(K_{2}\otimes H\right)=r+1$. Let $\alpha$ be an interval
$t$-coloring of the graph $G$, $\beta$ be an interval
$(r+1)$-coloring of the graph $K_{2}\otimes H$.

Define an edge coloring $\gamma$ of the graph $G\boxtimes H$ in the
following way:

\begin{description}
\item[(1)] for every $\left(w_{p}^{(i)},w_{q}^{(j)}\right)\in E(G\otimes H)$

\begin{center}
$\gamma
\left(\left(w_{p}^{(i)},w_{q}^{(j)}\right)\right)=\left(\alpha
\left(\left(u_{i},u_{j}\right)\right)-1\right)\cdot (r+1)+\beta
\left(\left(x_{p},y_{q}\right)\right)$,
\end{center}
where $1\leq i\leq n, 1\leq j\leq n, 1\leq p\leq m, 1\leq q\leq m$.

\item[(2)] for $i=1,2,\ldots,n$, the edges of the subgraph $H_{i}$ we color properly with colors
\begin{center}
$\max S\left(u_{i},\alpha\right)\cdot(r+1)+1,\max
S\left(u_{i},\alpha\right)\cdot(r+1)+2,\ldots,\max
S\left(u_{i},\alpha\right)\cdot(r+1)+r$
\end{center}
\end{description}\

It is easy to see that $\gamma$ is an interval $(t\cdot
(r+1)+r)$-coloring of the graph $G\boxtimes H$. By the definition of
$\gamma$, we have $w(G\boxtimes H)\leq w(G)\cdot (r+1)+r$ and
$W(G\boxtimes H)\geq W(G)\cdot (r+1)+r$. ~$\square$
\end{proof}

The Figure \ref{example2} shows the interval $11$-coloring $\gamma$
of the graph $P_{4}\boxtimes C_{4}$ described in the proof of
Theorem \ref{mytheorem14}.\\

\begin{figure}[h]
\begin{center}
\includegraphics[height=20pc,width=25pc]{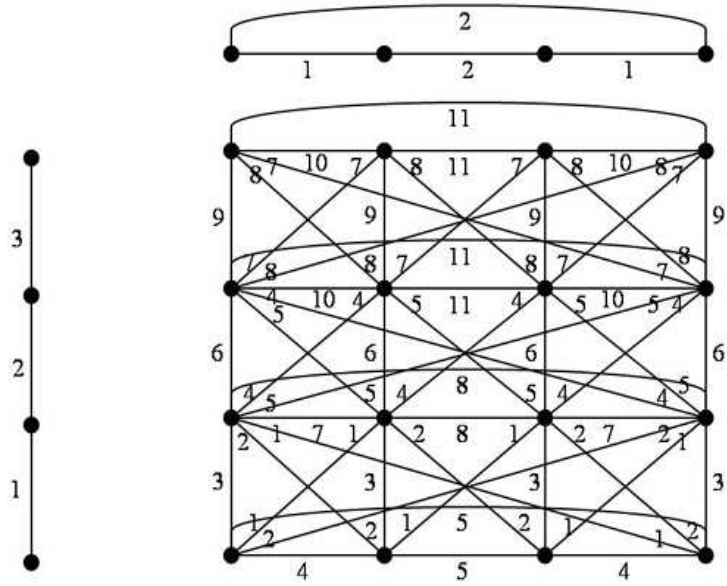}\
\caption{The interval $11$-coloring $\gamma$ of the graph
$P_{4}\boxtimes C_{4}$.}\label{example2}
\end{center}
\end{figure}

Note that there are graphs $G$ and $H$ for which $G\boxtimes H\in
\mathfrak{N}$, but $G\in \mathfrak{N},H\notin\mathfrak{N}$. For
example, $K_{2}\boxtimes C_{3}\in \mathfrak{N}$, but $C_{3}\notin
\mathfrak{N}$. For regular graphs the following result was obtained
by Zhou \cite{b38}.

\begin{theorem} \label{mytheorem15} If $G$ is $1$-factorable and $H$ is a regular
graph, then $G\boxtimes H$ is also $1$-factorable.
\end{theorem}

\begin{corollary}
\label{mycorollary6} Let $G$ and $H$ be two regular graphs and $G\in
\mathfrak{N}$. Then $G\boxtimes H\in \mathfrak{N}$.
\end{corollary}

Finally, we turn our attention to interval edge colorings of the
lexicographic product of graphs. In \cite{b25} Kubale and Giaro
posed the following question:

\begin{problem}
Does $G[H]\in \mathfrak{N}$ if $G,H\in \mathfrak{N}$?
\end{problem}

We start by focusing on the special case of this problem, when $G\in
\mathfrak{N}$ and $H=nK_{1}$ for any $n\in \mathbf{N}$.

\begin{theorem}
\label{mytheorem16} If $G\in \mathfrak{N}$, then $G[nK_{1}]\in
\mathfrak{N}$ for any $n\in \mathbf{N}$. Moreover, $w(G[nK_{1}])\leq
w(G)\cdot n$ and $W(G[nK_{1}])\geq (W(G)+1)\cdot n - 1$.
\end{theorem}
\begin{proof} Let $V(G)=\{u_{1},u_{2},\ldots,u_{m}\}$ and

\begin{center}
$V\left(G[nK_{1}]\right)=\left\{v_{j}^{(i)}|~1\leq i\leq m, 1\leq
j\leq n\right\}$,
\end{center}

\begin{center}
$E\left(G[nK_{1}]\right)=\left\{\left(v_{p}^{(i)},v_{q}^{(j)}\right)|~\left(u_{i},u_{j}\right)\in
E(G)~and~p,q=1,2,\ldots,n\right\}$.
\end{center}

Let $\alpha$ be an interval $t$-coloring of the graph $G$.

Define an edge coloring $\beta$ of the graph $G[nK_{1}]$ in the
following way:

for every $\left(v_{p}^{(i)},v_{q}^{(j)}\right)\in
E\left(G[nK_{1}]\right)$

\begin{center}
$\beta\left(\left(v_{p}^{(i)},v_{q}^{(j)}\right)\right)= \left\{
\begin{tabular}{ll}
$\left(\alpha((u_{i},u_{j}))-1\right)\cdot n + p+q-1\pmod {n}$, if $p+q\neq n+1$,\\
$\alpha((u_{i},u_{j}))\cdot n$, if $p+q=n+1$.\\
\end{tabular}%
\right.$
\end{center}
where $1\leq i\leq m, 1\leq j\leq m, 1\leq p\leq n, 1\leq q\leq
n$.\\

It can be verified that $\beta$ is an interval $t\cdot n$-coloring
of the graph $G[nK_{1}]$. By the definition of $\beta$, we have
$w(G[nK_{1}])\leq w(G)\cdot n$.

Now we show that $W(G[nK_{1}])\geq (W(G)+1)\cdot n - 1$.

Let $\phi$ be an interval $W(G)$-coloring of the graph $G$.

Define an edge coloring $\psi$ of the graph $G[nK_{1}]$ in the
following way:

for every $\left(v_{p}^{(i)},v_{q}^{(j)}\right)\in
E\left(G[nK_{1}]\right)$

\begin{center}
$\psi\left(\left(v_{p}^{(i)},v_{q}^{(j)}\right)\right)=
\left(\phi((u_{i},u_{j}))-1\right)\cdot n + p+q-1$,
\end{center}
where $1\leq i\leq m, 1\leq j\leq m, 1\leq p\leq n, 1\leq q\leq
n$.\\

It is easy to see that $\psi$ is an interval $(W(G)\cdot
n+n-1)$-coloring of the graph $G[nK_{1}]$. ~$\square$
\end{proof}

The Figure \ref{example3} shows the interval $6$-coloring $\beta$ of
the graph $(K_{1,3}+e)[2K_{1}]$ described in the proof of Theorem
\ref{mytheorem16}.\\

\begin{figure}[h]
\begin{center}
\includegraphics[height=20pc,width=30pc]{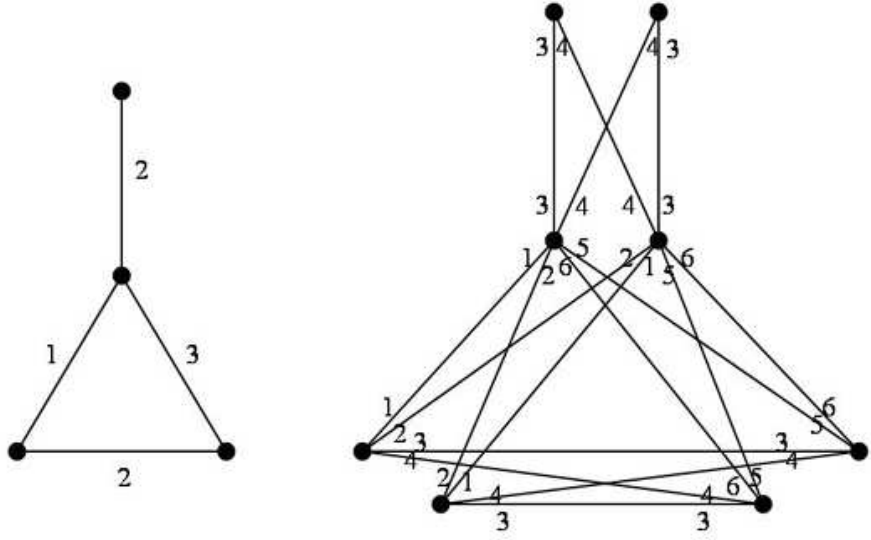}\
\caption{The interval $6$-coloring $\beta$ of the graph
$(K_{1,3}+e)[2K_{1}]$.}\label{example3}
\end{center}
\end{figure}

\begin{corollary}
\label{mycorollary7} (Kamalian, Petrosyan \cite{b22}) If $k$ is
even, then $C_{k}[nK_{1}]\in \mathfrak{N}$ and
\begin{center}
$W(C_{k}[nK_{1}])\geq 2n+\frac{n\cdot k}{2}-1$.
\end{center}
\end{corollary}

\begin{corollary}
\label{mycorollary8} (Kamalian, Petrosyan \cite{b23}) Let
$k=p2^{q}$, where $p$ is odd and $q\in \mathbf{N}$. Then
$K_{k}[nK_{1}]\in \mathfrak{N}$ and
\begin{center}
$W(K_{k}[nK_{1}])\geq (2k-p-q)\cdot n - 1$.
\end{center}
\end{corollary}

Now we show that $G[H]\in \mathfrak{N}$ if $G,H\in \mathfrak{N}$ and
$H$ is regular.

\begin{theorem}
\label{mytheorem17} If $G,H\in \mathfrak{N}$ and $H$ is an
$r$-regular graph, then $G[H]\in \mathfrak{N}$. Moreover, if $\vert
V(H)\vert= n$, then $w(G[H])\leq w(G)\cdot n + r$ and $W(G[H])\geq
W(G)\cdot n + r$.
\end{theorem}
\begin{proof} Let $V(G)=\{u_{1},u_{2},\ldots,u_{m}\}$, $V(H)=\{v_{1},v_{2},\ldots,v_{n}\}$  and

\begin{center}
$V\left(G[H]\right)=\bigcup_{i=1}^{m}V^{i}(H)$, where
$V^{i}(H)=\{w_{j}^{(i)}|~1\leq j\leq n\}$,
\end{center}

\begin{center}
$E\left(G[H]\right)=\left\{\left(w_{p}^{(i)},w_{q}^{(j)}\right)|~\left(u_{i},u_{j}\right)\in
E(G)~and~p,q=1,2,\ldots,n\right\}\cup\bigcup_{i=1}^{m}E^{i}(H)$,
where
$E^{i}(H)=\left\{\left(w_{p}^{(i)},w_{q}^{(i)}\right)|~\left(v_{p},v_{q}\right)\in
E(H)\right\}$.
\end{center}

Let $\alpha$ be an interval $t$-coloring of the graph $G$ and
\begin{center}
$H_{i}=\left(V^{i}(H),E^{i}(H)\right)$ for $i=1,2,\ldots,m$.
\end{center}

Note that $\chi ^{\prime}(H)=\Delta(H)=r$ since $H\in \mathfrak{N}$
and $H$ is an $r$-regular graph. This implies that there exists an
interval $r$-coloring of the graph $H$.

Define an edge coloring $\beta$ of the graph $G[H]$ in the following
way:

\begin{description}
\item[(1)] for every $\left(w_{p}^{(i)},w_{q}^{(j)}\right)\in E(G[H])$

\begin{center}
$\beta\left(\left(w_{p}^{(i)},w_{q}^{(j)}\right)\right)= \left\{
\begin{tabular}{ll}
$r+\left(\alpha((u_{i},u_{j}))-1\right)\cdot n + p+q-1\pmod {n}$,\\~~~~~~~~~~~~~~~~~~~~~~~~if $p+q\neq n+1$,\\
$r+\alpha((u_{i},u_{j}))\cdot n$, if $p+q=n+1$,\\
\end{tabular}%
\right.$
\end{center}
where $1\leq i\leq m, 1\leq j\leq m,i\neq j, 1\leq p\leq n, 1\leq
q\leq n$.\\

\item[(2)] for $i=1,2,\ldots,m$, the edges of the subgraph $H_{i}$ we color properly with colors
\begin{center}
$\left(\min S\left(u_{i},\alpha\right)-1\right)\cdot n+1,\left(\min
S\left(u_{i},\alpha\right)-1\right)\cdot n+2,\ldots,\left(\min
S\left(u_{i},\alpha\right)-1\right)\cdot n+r$
\end{center}
\end{description}\

It can be verified that $\beta$ is an interval $(t\cdot n
+r)$-coloring of the graph $G[H]$. By the definition of $\beta$, we
have $w(G[H])\leq w(G)\cdot n + r$ and $W(G[H])\geq W(G)\cdot n +
r$. ~$\square$
\end{proof}

The Figure \ref{example4} shows the interval $9$-coloring $\beta$ of
the graph $K_{4}[K_{2}]$ described in the proof of Theorem
\ref{mytheorem17}.\\

\begin{figure}[h]
\begin{center}
\includegraphics[height=20pc,width=35pc]{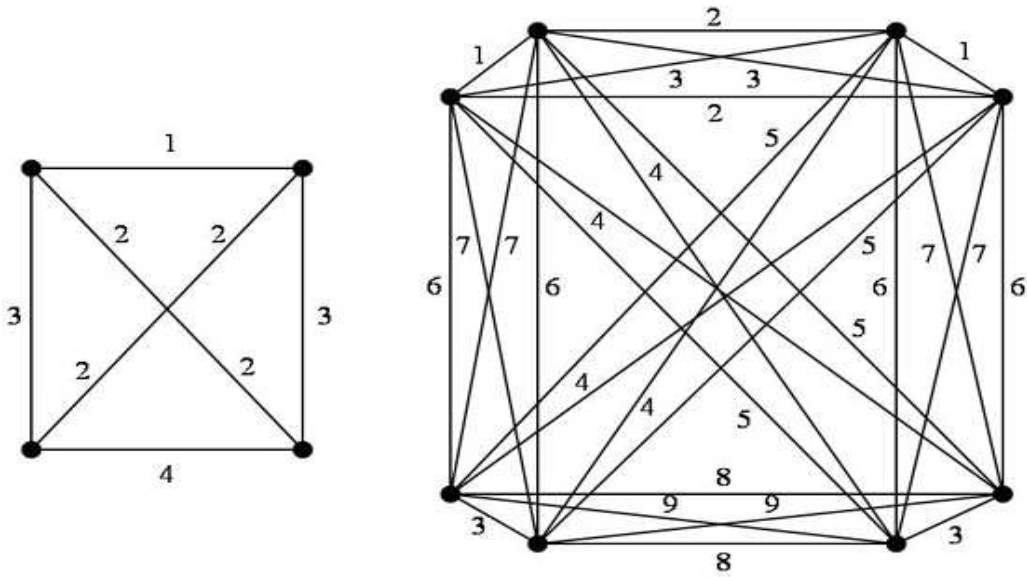}\
\caption{The interval $9$-coloring $\beta$ of the graph
$K_{4}[K_{2}]$.}\label{example4}
\end{center}
\end{figure}

\bigskip

\section{Problems}\

We conclude with the following problems on interval edge colorings
of products of graphs.

\begin{problem}
Are there graphs $G,H\notin \mathfrak{N}$, such that $G\times H\in
\mathfrak{N}$?
\end{problem}

\begin{problem}
Are there graphs $G,H\notin \mathfrak{N}$, such that $G\otimes H\in
\mathfrak{N}$?
\end{problem}

\begin{problem}
Are there graphs $G,H\notin \mathfrak{N}$, such that $G\boxtimes
H\in \mathfrak{N}$?
\end{problem}

\bigskip

\begin{acknowledgement}
We would like to thank the anonymous referees for useful
suggestions.
\end{acknowledgement}


\begin{thebibliography}{99}

\bibitem{b1} A.S. Asratian, R.R. Kamalian, Interval colorings of edges of a
multigraph, Appl. Math. 5 (1987) 25-34 (in Russian).

\bibitem{b2} A.S. Asratian, R.R. Kamalian, Investigation on interval
edge-colorings of graphs, J. Combin. Theory Ser. B 62 (1994) 34-43.

\bibitem{b3} A.S. Asratian, T.M.J. Denley, R. Haggkvist, Bipartite Graphs and
their Applications, Cambridge University Press, Cambridge, 1998.

\bibitem{b4} A.S. Asratian, C.J. Casselgren, J. Vandenbussche, D.B. West,
Proper path-factors and interval edge-coloring of $\left(3,4\right)
$-biregular bigraphs, J. Graph Theory 61 (2009) 88-97.

\bibitem{b5} M.A. Axenovich, On interval colorings of planar graphs, Congr.
Numer. 159 (2002) 77-94.

\bibitem{b6} C. Berge, Theorie des Graphes et ses Applications, Dunod, Paris,
1958 (in French).

\bibitem{b7} M. Bouchard, A. Hertz, G. Desaulniers, Lower bounds and a tabu
search algorithm for the minimum deficiency problem, J. Comb. Optim.
17 (2009) 168-191.

\bibitem{b8} Y. Feng, Q. Huang, Consecutive edge-coloring of the generalized $\theta$-graph,
Discrete Appl. Math. 155 (2007) 2321-2327.

\bibitem{b9} K. Giaro, The complexity of consecutive $\Delta $-coloring of
bipartite graphs: $4$ is easy, $5$ is hard, Ars Combin. 47 (1997)
287-298.

\bibitem{b10} K. Giaro, M. Kubale, Consecutive edge-colorings of complete and
incomplete Cartesian products of graphs, Congr, Numer. 128 (1997)
143-149.

\bibitem{b11} K. Giaro, M. Kubale, M. Malafiejski, On the deficiency of
bipartite graphs, Discrete Appl. Math. 94 (1999) 193-203.

\bibitem{b12} K. Giaro, M. Kubale, M. Malafiejski, Consecutive colorings of
the edges of general graphs, Discrete Math. 236 (2001) 131-143.

\bibitem{b13} K. Giaro, M. Kubale, Compact scheduling of zero-one time
operations in multi-stage systems, Discrete Appl. Math. 145 (2004)
95-103.

\bibitem{b14} H.M. Hansen, Scheduling with minimum waiting periods, Master's
Thesis, Odense University, Odense, Denmark, 1992 (in Danish).

\bibitem{b15} D. Hanson, C.O.M. Loten, A lower bound for interval coloring of
bi-regular bipartite graphs, Bull. ICA 18 (1996) 69-74.

\bibitem{b16} D. Hanson, C.O.M. Loten, B. Toft, On interval colorings of
bi-regular bipartite graphs, Ars Combin. 50 (1998) 23-32.

\bibitem{b17} P.E. Himmelwright, J.E. Williamson, On 1-factorability and
edge-colorability of cartesian products of graphs, Elem. Der Math.
29 (1974) 66-67.

\bibitem{b18} T.R. Jensen, B. Toft, Graph Coloring Problems, Wiley
Interscience Series in Discrete Mathematics and Optimization, 1995.

\bibitem{b19} R.R. Kamalian, Interval colorings of complete bipartite graphs
and trees, preprint, Comp. Cen. of Acad. Sci. of Armenian SSR,
Erevan, 1989 (in Russian).

\bibitem{b20} R.R. Kamalian, Interval edge colorings of graphs, Doctoral
Thesis, Novosibirsk, 1990.

\bibitem{b21} R.R. Kamalian, A.N. Mirumian, Interval edge colorings of
bipartite graphs of some class, Dokl. NAN RA, 97 (1997) 3-5 (in
Russian).

\bibitem{b22} R.R. Kamalian, P.A. Petrosyan, Interval colorings of some
regular graphs, Math. probl. of comp. sci. 25 (2006) 53-56.

\bibitem{b23} R.R. Kamalian, P.A. Petrosyan, On interval colorings of
complete $k$-partite graphs $K_{n}^{k}$, Math. probl. of comp. sci.
26 (2006) 28-32.

\bibitem{b24} A. Kotzig, 1-Factorizations of cartesian products of regular graphs,
J. Graph Theory 3 (1979) 23-34.

\bibitem{b25} M. Kubale, Graph Colorings, American Mathematical Society, 2004.

\bibitem{b26} E.S. Mahamoodian, On edge-colorability of cartesian products of graphs,
Canad. Math. Bull. 24 (1981) 107-108.

\bibitem{b27} B. Mohar, T. Pisanski, Edge-coloring of a family of regular graphs,
Publ. Inst. Math. (Beograd) 33 (47) (1983) 157-162.

\bibitem{b28} B. Mohar, On edge-colorability of products of graphs,
Publ. Inst. Math. (Beograd) 36 (50) (1984) 13-16.

\bibitem{b29} P.A. Petrosyan, G.H. Karapetyan, Lower bounds for the greatest
possible number of colors in interval edge colorings of bipartite
cylinders and bipartite tori, Proceedings of the CSIT Conference
(2007) 86-88.

\bibitem{b30} P.A. Petrosyan, Interval edge-colorings of complete graphs and
$n$-dimensional cubes, Discrete Mathematics 310 (2010) 1580-1587.

\bibitem{b31} T. Pisanski, J. Shawe-Taylor, B. Mohar, 1-Factorization of the
composition of regular graphs, Publ. Inst. Math. (Beograd) 33 (47)
(1983) 193-196.

\bibitem{b32} A.V. Pyatkin, Interval coloring of $\left( 3,4\right)
$-biregular bipartite graphs having large cubic subgraphs, J. Graph
Theory 47 (2004) 122-128.

\bibitem{b33} G. Sabidussi, Graph multiplication, Math. Z. 72 (1960) 446-457.

\bibitem{b34} A. Schwartz, The deficiency of a regular graph, Discrete Math.
306 (2006) 1947-1954.

\bibitem{b35} S.V. Sevast'janov, Interval colorability of the edges of a
bipartite graph, Metody Diskret. Analiza 50 (1990) 61-72 (in
Russian).

\bibitem{b36} V.G. Vizing, The Cartesian product of graphs, Vych. Sis. 9 (1963) 30-43 (in
Russian).

\bibitem{b37} D.B. West, Introduction to Graph Theory, Prentice-Hall, New
Jersey, 1996.

\bibitem{b38} M.K. Zhou, Decomposition of some product graphs into 1-factors and Hamiltonian cycles,
Ars Combin. 28 (1989) 258-268.

\end{thebibliography}
\end{document}